\documentclass[letterpaper, 10 pt, conference]{ieeeconf} 
\makeatletter
\def\endfigure{\end@float}
\def\endtable{\end@float}
\makeatother

\usepackage{algorithm}
\usepackage{algpseudocode}
\usepackage{amsmath,amssymb,amsthm}
\usepackage{graphicx}
\usepackage{booktabs}
\usepackage{placeins}
\usepackage{multirow}
\usepackage[caption=false,font=footnotesize]{subfig}
\usepackage{balance}
\usepackage{hyperref}
\usepackage{todonotes}
\usepackage{pifont}



\usepackage{etoolbox}
\pretocmd{\printbibliography}{\FloatBarrier\balance\sloppy}{}{}

\newtheorem{assumption}{Assumption}
\newtheorem{definition}{Definition}

\newtheorem{theorem}{Theorem}

\newtheorem{remark}{Remark}
\newtheorem{corollary}{Corollary}

\usepackage{ifthen}
\newboolean{showcomments}
\setboolean{showcomments}{true}
\usepackage{color}
\usepackage{todonotes}

\definecolor{bleudefrance}{rgb}{0.19, 0.55, 0.91}
\definecolor{ao(english)}{rgb}{0.0, 0.5, 0.0}

\newcommand{\addcite}[0]{\ifthenelse{\boolean{showcomments}}
{\textcolor{purple}{(add cite(s)) }}{}}%

\newcommand{\addcites}[0]{\ifthenelse{\boolean{showcomments}}
{\textcolor{purple}{(add cite(s)) }}{}}%

\newcommand{\addref}[0]{\ifthenelse{\boolean{showcomments}}
{\textcolor{purple}{(add ref) }}{}}%

\newcommand{\enrique}[1]{  \ifthenelse{\boolean{showcomments}}
{\todo[inline,color=bleudefrance]{Enrique: #1}}{}}
\newcommand{\jixian}[1]{  \ifthenelse{\boolean{showcomments}}
{\todo[inline,color=lightgray]{Jixian: #1}}{}}

\newcommand{\hl}[1]{\ifthenelse{\boolean{showcomments}}
{\textcolor{red}{#1}}{#1}}

\newboolean{showedits}
\setboolean{showedits}{false}
\usepackage[markup=underlined]{changes}
\definechangesauthor[color=bleudefrance]{EM}
\newcommand{\aem}[1]{
\ifthenelse{\boolean{showedits}}
{\added[id=EM]{#1}}
{\!#1\hspace{-4.75pt}}
}
\newcommand{\repem}[2]{
\ifthenelse{\boolean{showedits}}
{\replaced[id=EM]{#1}{#2}}
{\!#1\hspace{-4.75pt}}
}
\newcommand{\dem}[1]{
\ifthenelse{\boolean{showedits}}
{\deleted[id=EM]{#1}}
{}
}

\usepackage[backend=biber, style=ieee, doi=false, url=false, eprint=false]{biblatex}
\newcommand{\md}{\mathrm{d}}
\newcommand{\cl}{\mathrm{cl}}

\newcommand{\K}{\mathcal{K}}
\newcommand{\X}{\mathcal{X}}
\newcommand{\U}{\mathcal{U}}
\newcommand{\R}{\mathbb{R}}
\newcommand{\s}{\mathcal{S}}
\newcommand{\Z}{\mathcal{Z}}

\newcommand{\Oi}{\mathcal{O}_i}
\newcommand{\Of}{\mathcal{O}}
\allowdisplaybreaks

\let\labelindent\relax
\usepackage{enumitem}
\IEEEoverridecommandlockouts
\title{\LARGE \bf 
Safety-Critical Control via Recurrent Tracking Functions}

\author{Jixian Liu and Enrique Mallada
\thanks{J. Liu and E. Mallada are with the Department of Electrical and Computer Engineering, Johns Hopkins University, MD 21218, U.S.A. 
        {\tt\small jliu376@jh.edu,mallada@jhu.edu}.}
\thanks{This work was supported by the NSF Global Centers program under Grant No.~2330450 and by the DOE Office of Science (ASCR) under Award No.~826565.}}

\begin{document}
\maketitle
\thispagestyle{empty}
\pagestyle{empty}

\begin{abstract}
This paper addresses the challenge of synthesizing safety-critical controllers for high-order nonlinear systems, where constructing valid Control Barrier Functions (CBFs) remains computationally intractable. Leveraging layered control, we design CBFs in reduced-order models (RoMs) while regulating full-order models' (FoMs) dynamics at the same time. Traditional Lyapunov tracking functions are required to decrease monotonically, and systematic synthesis methods for such functions exist only for fully-actuated systems. To overcome this limitation, we introduce Recurrent Tracking Functions (RTFs), which replace the monotonic decay requirement with a weaker finite-time recurrence condition. This relaxation permits transient deviations of tracking errors while ensuring safety. By integrating CBFs for RoMs with RTFs, we construct recurrent CBFs (RCBFs) whose zero-superlevel set is control $\tau$-recurrent, and guarantee safety for all initial states in such a set when RTFs are satisfied. We establish theoretical safety guarantees and validate the approach through a proof-of-concept numerical experiment, demonstrating RTFs' effectiveness and the safety of FoMs.
\end{abstract}

\section{Introduction}
Safety-Critical Control in autonomous systems requires controllers that can guarantee constraint satisfaction while achieving performance objectives. Control Barrier Functions (CBFs) provide a principled framework for certifying safety by rendering desired safe sets forward invariant through appropriate control actions. However, the practical deployment of CBF-based methods hinges on the ability to construct valid barrier functions for the system at hand---a challenge that grows increasingly difficult as system complexity and dimensionality increase~\cite{compton2024learning}.

Existing CBF synthesis methods typically leverage certain relative degree properties~\cite{nguyen2016exponential,wang2021learning}, learning-based approaches~\cite{robey2021learning,srinivasan2020synthesis}, sum-of-squares programs for polynomial dynamics~\cite{clark2021verification, dai2022convex}, or reachability-based constructions from Hamilton--Jacobi value functions~\cite{choi2021robust}. These approaches often rely on assumptions that are difficult to satisfy in practice or are hampered by the curse of dimensionality, which limits their applicability to complex systems. As a result, scalable CBF synthesis for systems with higher-order dynamics, such as aerial vehicles, legged robots, and complex power grids, remains an important open problem~\cite{cohen2024safety}.

A promising alternative is layered control: a reduced-order model (RoM) captures safety-critical states (e.g., collision-relevant kinematics or center-of-mass motion) while a low-level controller regulates the full-order model (FoM)~\cite{matni2024quantitative,molnar2021model}. Designing CBFs in the lower-dimensional RoM reduces computation, and safety transfers to the FoM when tracking errors between the true signals and the safe ones admit exponential convergence certificates~\cite {molnar2023safety}. However, beyond fully actuated systems, there is no systematic procedure for constructing such tracking functions to certify the Lyapunov condition for exponential convergence, which hinders practice. Robust RoM-based CBFs with predictive error margins~\cite{compton2024learning} mitigate RoM–FoM mismatch, but the conservatism term~$\delta$ discards useful system information.

In this paper, we address these issues by relaxing the Lyapunov condition to a recurrent one~\cite{sspm2023cdc,ssm2024allerton,sibai2026recurrence,liu2025RCBF}: instead of enforcing monotone decay of the tracking function, we only require the decrease condition to recurrently hold on time instances that are separated by at most a uniform bound $\tau$.

Concretely, we introduce the concept of \textbf{Recurrent Tracking Functions (RTFs)}, for which any norm is a valid choice, thereby providing a flexible RoM--FoM linkage for general nonlinear systems. Then augmenting a RoM CBF with an RTF yields a recurrent CBF $h_V$~\cite{liu2025RCBF}, whose zero-superlevel set $S_V$ is control $\tau$-recurrent: trajectories may temporarily leave $S_V$ for at most $\tau$ time units, but never enter the unsafe region. Our main contributions are as follows:
\begin{enumerate}
    \item \textbf{RTF framework:} We formulate RTFs as a relaxation of standard Lyapunov decrease conditions. For Lipschitz continuous systems with exponentially stable tracking error, any norm of the tracking error is a valid RTF~\cite{sspm2025tac}.
    \item \textbf{Safety analysis:} We prove that if the initial state lies in $S_V$, then any controller satisfying the RTF condition guarantees safety for all time.
    \item \textbf{Numerical validation:} We demonstrate the effectiveness of the proposed approach on the FoM through numerical experiments.
\end{enumerate}

\noindent\textbf{Organization.} Section~\ref{sec:prelims} reviews layered control, CBFs, and recurrence-based analysis. Section~\ref{sec:RTF} presents the proposed RTF-based framework and its associated safety analysis. Section~\ref{sec:simulations} presents numerical results. Section~\ref{sec:conclusion} concludes the paper and outlines future directions.

\section{Preliminaries and Problem Formulation}
\label{sec:prelims}
\subsection{Layered Control System}
We begin this paper by recalling the setting of a layered-control system, consisting of a \textbf{full-order model (FoM)} and a \textbf{reduced-order model (RoM)}, which interacts with the FoM through the projection of the full state onto the RoM~\cite{matni2024quantitative}. 

\begin{definition}[Full-Order Model]
The FoM represents the physical system, typically high-dimensional, nonlinear, and only partially known. Its dynamics are assumed to evolve as
\begin{align}
    \dot{x} = F(x,u),
    \label{eq:Control_System}
\end{align}
where $x \in \mathcal{X} \subseteq \R^{N}$ is the state and $u \in U \subseteq \R^M$ is the control input.
\end{definition}

Given an interval $\mathcal{I}\subseteq\R_{\geq0}$, 
let $\mathcal{U}^{\mathcal{I}} := \{u:\mathcal{I}\to U \mid u \text{ measurable}\}$. We use $u$ for both instantaneous inputs $u\in U$ and signals $u\in\mathcal{U}^{\mathcal{I}}$, disambiguated by context, and $\mathcal{U}:=\mathcal{U}^{\R_{\geq0}}$. Given initial condition $x_0 \in \X$ and $u \in \mathcal{U}^{(0,a]}$, let $\phi(t,x_0,u)$ be the trajectory of~\eqref{eq:Control_System} at time $t\in(0,a]$.

\begin{assumption}[Forward Completeness]\label{ass:forward}
The control system~\eqref{eq:Control_System} is forward complete: for every initial condition $x_0 \in \mathcal{X} \subseteq \R^N$ and every admissible input signal $u \in \mathcal{U}$, the solution $\phi(t,x_0,u)$ of~\eqref{eq:Control_System}
exists for all $t \ge 0$. 
\end{assumption}

\begin{assumption}[Uniform Local Lipschitz Continuity]\label{ass:lipschitz}
The vector field $F:\mathcal{X} \times U \to \R^N$ in~\eqref{eq:Control_System} is continuous and locally Lipschitz in $x$, uniformly w.r.t. $u$, i.e., for every compact $S \subseteq \mathcal{X}$, there exists $L_{F_{S}} \ge 0$ such that 
\begin{align*}
\|F(y,u)-F(x,u)\| \le L_{F_{S}}\|y-x\|, \forall x,y \in S, \forall u \in U.
\end{align*}
\end{assumption}

\begin{definition}[Reduced-order Model]\label{def:rom}
The RoM captures safety-relevant, lower-dimensional dynamics:
\begin{align}
\label{eq:RoM}
    \dot z = f(z,v) ,
\end{align}
with state $z \in \Z \subseteq \R^{n}$ $(n<N),$ and input $v \in \mathrm{V} \subseteq \R^{m}$ $(m<M)$, where $f$ is also locally Lipschitz continuous with constant $L_{R_{S}}$ for every compact set $S\subseteq \Z$.
\end{definition}

And a closed ball is defined as
\begin{definition}[Closed Ball]
For any $y \in \R^d$ and $r>0$, define the closed ball centered at $y$ with radius $r$ as
\begin{align}
\mathcal{B}_r(y) := \{ w \in \R^d \mid \|w-y\| \le r \}.
\end{align}
\end{definition}

\noindent
\textbf{RoM-FoM Interconnection.}
The RoM and FoM are related through a projection map $\Pi:\mathcal X\to\mathcal Z$ and a tracking interface $K:\mathcal X\times V\to U$. Let $k:\mathcal Z\to V$ denote the feedback law of the RoM. Given a full-order state $x\in\mathcal X$, the projected state $z=\Pi(x)$ is used to generate the reference input $v = k(z)$, which is then passed to the FoM controller $u = K(x,v)$. This induces the closed-loop FoM
\begin{align}
\label{eq:closed_loop_FoM}
    \dot{x} = F_{\cl}(x) := F\bigl(x,K(x,k(\Pi(x)))\bigr).
\end{align}
We denote by $x_{\cl}(t)$ the corresponding closed-loop state trajectory with initial condition $x_0$, i.e., $x_\cl(0)=x_0$.

\subsection{Problem Statement}
To formalize the safety objective, let $h:\mathcal Z\to\R$ be a smooth function that defines the safe set of the reduced-order model (RoM).

\begin{definition}[Safe State of the RoM]\label{def:safe-rom}
A state $z \in \Z \subseteq \R^n$ is \textbf{safe} if $z\in \s_{\mathrm{RoM}}$, where
\begin{align}
    \s_{\mathrm{RoM}} := \{ z \in \Z : h(z) \ge 0 \}
\end{align}
is the RoM's safe set with boundary $\partial \s_{\mathrm{RoM}} = h^{-1}(0)$.
\end{definition}

\begin{definition}[Safe State of the FoM]\label{def:safe-fom}
A state $x \in \mathcal{X} \subseteq \R^N$ is \textbf{safe} if $x \in \s_{\mathrm{FoM}}$, where
\begin{align}
    \s_{\mathrm{FoM}} := \{ x \in \mathcal{X} : h(\Pi(x)) \ge 0 \}
\end{align}
is the FoM's safe set with boundary $\partial \s_{\mathrm{FoM}} = \Pi^{-1}(h^{-1}(0))$.
\end{definition}

The objective is to guarantee the safety of the FoM by synthesizing a controller based on the safe reference generated by the RoM.

\noindent
\textbf{Problem.} Given the RoM dynamics~\eqref{eq:RoM}, the closed-loop FoM~\eqref{eq:closed_loop_FoM}, and the safe set $\mathcal{S}_{\mathrm{FoM}}$, design a RoM feedback law $v = k(\Pi(x))$ and a compact set $S \subseteq \mathcal{S}_{\mathrm{FoM}}$ such that $x \in S \Rightarrow x_{\cl}(t) \in \mathcal{S}_{\mathrm{FoM}}, \forall t \ge 0.$

\subsection{Control Barrier Functions for the RoM}
Control barrier functions (CBFs) provide a standard tool for certifying safety of the RoM~\eqref{eq:RoM}. In particular, if a smooth function $h$ satisfies a suitable differential inequality involving an extended class-$\K$ function, then its zero-superlevel set is forward invariant and therefore safe. We first recall the notion of an extended class-$\K$ function.

\begin{definition}[Extended Class-$\K$ Function]
A function $\kappa: \R \to \R$ is an extended class $\K$ function if it is continuous, strictly increasing, and satisfies $\kappa(0) = 0$.
\end{definition}

With this notion in place, we can formalize the CBF condition used throughout the paper.

\begin{definition}[Control Barrier Function~\cite{acenst2019ecc}]\label{def:CBF}
A continuously differentiable function $h:\mathcal Z\to\R$ is a CBF for~\eqref{eq:RoM} if there exists an extended class-$\K$ function $\kappa$ such that
\begin{align}\label{eq:CBF}
    \max_{v \in V} \Bigl( L_f h(z,v) + \kappa(h(z)) \Bigr) \ge 0,
\end{align}
for all $z \in \mathcal Z$, where $
L_f h(z,v) := \frac{\partial h}{\partial z}(z)^\top f(z,v)$ denotes the Lie derivative of $h$ along the RoM dynamics.
\end{definition}

\begin{theorem}[\cite{acenst2019ecc}]
As a direct consequence of Definition~\ref{def:CBF}, any Lipschitz-continuous controller $k(z)$ that satisfies
\begin{align}
k(z) \in \{v \in V \mid L_{f}h(z) + \kappa(h(z)) \ge 0\}
\end{align}
renders the set $h_{\ge 0}:=\{z \mid h(z) \ge 0\}$ forward invariant. In particular, $h_{\ge0}$ is control invariant.
\end{theorem}

\subsection{Recurrent Lyapunov and Control Barrier Functions}
Constructing a Lyapunov function for a general high-order system is often intractable. Reference~\cite{sspm2025tac} relaxes this requirement via Recurrent Lyapunov Functions (RLFs), which require the Lyapunov decrease condition to hold only recurrently rather than continuously. RLFs still guarantee exponential stability, and a converse theorem shows that any norm of the state satisfies the RLF conditions for exponentially stable systems. Since RLF conditions are defined for autonomous systems, we express the results in this section in terms of $x_{\cl}(t)$, which denotes solutions to~\eqref{eq:closed_loop_FoM}.

\begin{definition}[Exponential Stability]\label{def:exp_stab}
Given $S\subseteq\R^N$, an equilibrium point $x^*$ of~\eqref{eq:closed_loop_FoM} is \emph{exponentially stable} on $S$ if there exist constants $M,\lambda>0$ such that, for all $x_0\in S$,
\begin{equation}\label{eq:exp_stab}
    \|x_{\cl}(t)-x^*\| \le M e^{-\lambda t}\|x_0-x^*\|,\qquad \forall t\ge 0.
\end{equation}
\end{definition}

\begin{definition}[Reachable Tube]\label{def:reachable_tube}
For system~\eqref{eq:closed_loop_FoM}, $\tau>0$, and $S\subset \R^N$, the $\tau$-reachable tube of $S$ is
\begin{equation}
\mathcal{R}^{\tau}(S) = \bigcup_{\substack{x_0 \in S\\ t\in[0,\tau]}} \{x_{\cl}(t)\}.
\end{equation}
\end{definition}

\begin{definition}[Containment Times]\label{def:containment_times}
For $S\subseteq \R^N$, $x_0\in\X$, the containment time for~\eqref{eq:closed_loop_FoM} is defined as 
$
T_S(x_0) := \{t>0 \mid x_{\cl}(t)\in S\}.$
For constants $a,b$, $T_S(x_0;a,b) := T_S(x_0)\cap (a,a+b],$ and $ T_S(x_0; b):=T_S(x_0;0,b).$
\end{definition}

\begin{definition}[Recurrent Lyapunov Function~\cite{sspm2025tac}]\label{def:RCLF}
Given an equilibrium $x^*\in\mathcal{X}$ of~\eqref{eq:closed_loop_FoM} and a compact $S\subseteq\mathcal{X}$ with $x^*\in\mathrm{int}(S)$, a continuous $V:\mathcal{X}\to\R_{\ge0}$ is a Recurrent Lyapunov Function (RLF) over $S$ if:
\begin{enumerate}
\item (\textbf{Positive definiteness}) $\exists\,a_1,a_2>0$ such that
\begin{align}\label{eq:RLF_bound}
a_1 \|x - x^*\| \le V(x) \le a_2 \|x - x^*\|,\forall x \in S.
\end{align}
\item (\textbf{Exponential $\tau$-recurrence}) $\exists \alpha,\tau>0$ such that
\begin{align}\label{eq:RLF_decay}
\min_{t\in T_{S}(x_0;\tau)} e^{\alpha t} V(x_{\cl}(t)) \le V(x_0), \forall x_0 \in S.
\end{align}
\end{enumerate}
\end{definition}

Condition~\eqref{eq:RLF_decay} enforces exponential convergence via recurrent returns, relaxing monotone decrease.
Analogously, \cite{liu2025RCBF} defines a \textbf{Recurrent Control Barrier Function (RCBF)} $h^r:\R^N\to\R$, which replaces invariance with recurrence: trajectories may leave $h^r\ge0$ but must re-enter within time $\tau$, infinitely often; the resulting set is control $\tau$-recurrent.

\begin{remark}
The following requires extending the definition of containment times, i.e., Definition~\ref{def:containment_times}, for the FoM system~\eqref{eq:Control_System}.
For $S \subseteq \R^N, x_0 \in \R^{N}$, we use $T_{S}(x_0,u) := \{t> 0 \mid \phi(t, x_0, u) \in S\}$ denoted as the containment time for~\eqref{eq:Control_System} and $T_{S}(x_0, u; \tau) := T_{S}(x_0, u) \cap (0, \tau]$.
\end{remark}

\begin{definition}[Recurrent Control Barrier Function]\label{def:RCBF}
For~\eqref{eq:Control_System} and a compact $S\subseteq\R^N$, a continuous $h^r:\R^N\to\R$ is an RCBF over $S$ if for every $x\in S$ there exists $u \in \U^{(0,\tau]}$ such that:
\begin{align}\label{eq:RCBF}
\max_{t \in T_{S}(x,u;\tau)} e^{\gamma(h^r(\phi(t,x,u)))t} h^r(\phi(t,x,u)) \ge h^r(x),
\end{align}
where the function $\gamma: \R \to \R_{>0}$ is given.
\end{definition}

Unlike invariance-based safety, which requires the invariant set to be disjoint from unsafe regions, RCBFs certify safety if the recurrent set avoids the backward $\tau$-reachable tube of the unsafe set~\cite{bansal2017hamilton}. 

A key advantage of the recurrence framework is that both RLFs and RCBFs can be constructed using simple norm and signed distance functions, respectively~\cite{sspm2025tac,ssm2024allerton}, significantly simplifying synthesis compared to classical Lyapunov and barrier functions. 
This motivates replacing the exponential tracking certificates in~\cite{molnar2021model} with (possibly norm-based) Recurrent Tracking Certificates, yielding a systematic synthesis method aligned with~\eqref{eq:RLF_decay}--\eqref{eq:RCBF}.

\section{Layered Safety-Critical Control via Recurrence}\label{sec:RTF}
In this section, we extend the recurrence framework to provide safety guarantees for layered control architectures. A standard approach in this setting is to combine RoM safety with tracking-based interface assumptions to certify FoM safety~\cite{compton2024learning}.

\subsection{Ideal Tracking}
Firstly, we consider trajectories $x_{\cl}(t)$ of the closed loop system \eqref{eq:closed_loop_FoM}, and $z_{\cl}(t)$ is its projection, i.e., $z_\cl(t)=\Pi(x_\cl(t))$. Then consider a CBF $h: \R^n \rightarrow \R$ for the RoM~\eqref{eq:RoM} with linear extended class-$\K$ function $\kappa(h) = \alpha h$ for some $\alpha > 0$, and the safe velocity reference $\dot z_{s}(t)$ satisfies 
\begin{align}
\nabla h(z_{\cl}(\cdot))^\top \dot z_{s}(\cdot) \ge -\alpha h(z_{\cl}(\cdot)).
\end{align}
Moreover, along the projected closed-loop trajectory,
\begin{align}
\dot{h}(z_{\cl}(\cdot)) & = \nabla h(z_{\cl}(\cdot))^{\top}\dot{z_{\cl}}(\cdot) \notag\\
& = \nabla h(z_{\cl}(\cdot))^{\top}(\dot{z}_s(\cdot) + \dot{e}(\cdot)),
\end{align}
where the tracking error between the true velocity $\dot z_{\cl}$ and $\dot z_s$ is $\dot{e}(\cdot) := \dot{z}_{\cl}(\cdot) - \dot{z}_s(\cdot)$,
with initial condition $\dot{e}_0 := \dot{e}(0)$.

 We impose the following regularity condition on the RoM safety function to bound its sensitivity with respect to the projected FoM state.

\begin{assumption}[Boundedness]\label{ass:boundness}
There exists $C_h > 0$ s.t.
\begin{align*}
\|\nabla h(\Pi(x))\| \le C_h, \quad \forall x \in \mathcal{S}_{\mathrm{FoM}}.
\end{align*}
\end{assumption}

To capture the ideal tracking regime, we impose the following exponential tracking condition on the velocity error.

\begin{assumption}[Exponentially Stable Tracking]
\label{ass:Exponential_Tracking}
The velocity tracking error is assumed to decay exponentially, i.e., there exist positive constants M and $\beta$ such that
\begin{align}
\|\dot e(t)\| \le M \|\dot e_0\| e^{-\beta t}
\end{align}
for all $t \ge 0$.
\end{assumption}

 As shown later, this assumption provides a sufficient condition for constructing a valid recurrent tracking certificates.

\subsection{Recurrent Tracking Function}

We then formalize the notion of \textbf{Recurrent Tracking Functions (RTFs)} for system~\eqref{eq:closed_loop_FoM} by defining a condition akin to \eqref{eq:RLF_decay} for tracking errors. For a set $S \subseteq \mathcal{Z} \times \R^n$ and initial condition $x_0 \in \X$, where $(z_0, e_0) \in S$, define $T_S^{e}(x_0) := \{ t>0 \mid (z_{\cl}(t), \dot e(t)) \in S \},$
and $T_S^{e}(x_0;\tau) := T_S^{e}(x_0)\cap (0,\tau].$

\begin{definition}[Recurrent Tracking Function]\label{def:RTF}
Consider the system~\eqref{eq:closed_loop_FoM} and a compact set $S \subseteq \mathcal{Z} \times \R^n$ with $0 \in \mathrm{int}(S_{\dot e})$, where $S_{\dot e} := \{ \dot e \in \R^n : \exists z \in \R^n,\ (z,\dot e)\in S \}$. A continuous function $V:\mathcal{Z} \times \R^n \to \R_{\ge 0}$ is a \textbf{Recurrent Tracking Function (RTF)} over $S$ if:
\begin{enumerate}
    \item \textbf{Positive definiteness with linear error bounds:} There exist $a_1,a_2>0$ such that \begin{align}\label{eq:linear_contained}
        a_1 \|\dot e\| \le V(z, \dot e) \le a_2 \|\dot e\|, \forall (z,\dot e)\in S .
    \end{align}
    \item \textbf{$\beta$-exponential $\tau$-recurrence:} There exist $\tau,\beta>0$ such that for every $x\in\mathcal{X}$ s.t. $(z_0, \dot e_0) \in S$,  the corresponding tracking error $(z_{\cl}(\cdot),\dot e(\cdot))$
    satisfies
    \begin{align}\label{eq:exp_RTF}
        \min_{t \in T_{S}^e(x_0;\tau)} e^{\beta t} V(z_{\cl}(t), \dot e(t)) \le V(z_0, \dot e_0).
    \end{align}
\end{enumerate}
\end{definition}

In contrast to classical Lyapunov tracking functions, which must decrease monotonically along trajectories, RTF only require exponential $\tau$-recurrence. The exponential decay of the tracking error implies that $\|\dot e\|$ serve as a valid RTF. This avoids the need to explicitly construct a monotone tracking Lyapunov function, a task that is generally difficult beyond simple systems.

\begin{theorem}[Norm is Valid RTF]\label{theo:exponential_stability}
Consider system~\eqref{eq:closed_loop_FoM}, suppose Assumption~\ref{ass:Exponential_Tracking} holds, given a compact set $S \subseteq \mathcal{Z} \times \R^n$ with  $0 \in \mathrm{int}(S_{\dot e})$. then for any $0< \beta' <\beta$ and $\tau \ge \frac{1}{\beta - \beta'}\ln(M\frac{b_2}{b_1})$, where $\mathcal B_{b_1}(0) \subseteq S_{\dot e} \subseteq \mathcal B_{b_2}(0)$, the function $V(z,\dot e) := \|\dot e\|$ is a valid RTF where
\begin{align}
\min_{t \in T_{S}^e(x_0;\tau)} e^{\beta' t} V(z_{\cl}(t), \dot e(t)) \le V(z_0, \dot e_0).
\end{align}
\end{theorem}

\begin{proof}
The proof follows closely similar results for \cite[Theorem 6]{sspm2025tac}, and it is omitted due to space constraints.
\end{proof}

\subsection{Safety Assessment}
In this subsection, we combine a valid RTF for the closed-loop FoM~\eqref{eq:closed_loop_FoM} with a CBF for the RoM~\eqref{eq:RoM} to characterize the set of initial conditions from which safety of the closed-loop system can be guaranteed.

\begin{theorem}[Recurrent CBF Construction]\label{theo:RCBF_set}
Consider system~\eqref{eq:closed_loop_FoM} and suppose $h: \R^n \to \R$ is a CBF for the RoM~\eqref{eq:RoM} with extended class-$\mathcal{K}$ function $\kappa(h) = \alpha h$ for some $\alpha > 0$. If $V(z, \dot{e})$ is a valid RTF given recurrent time $\tau$ over $S \subseteq \mathcal{Z} \times \R^{n}$ with convergence rate $\beta > \alpha$, then 
\begin{align}
\label{eq:h_V}
    h_{V}(z,\dot{e}) = -V(z, \dot{e}) + \alpha_e h(z)
\end{align}
is an RCBF, where $\alpha_e = \frac{a_1^2(\beta - \alpha)}{a_2C_{h} M}$, and its zero-superlevel set
\begin{align}
    S_V := \{(z,\dot{e}) \in S : h_V(z,\dot{e}) \geq 0\}
\end{align}
is a control $\tau$-recurrent set. 

In particular, for any $x_0 \in \s_{\mathrm{FoM}}$ s.t. $(\Pi(x_0),\dot e_0)\in S_V$, the induced control signal $u(\cdot)=K(x(\cdot),k(\Pi(x(\cdot))))$ satisfies the RCBF condition.
\end{theorem}

\begin{proof} 
See Appendix~\ref{app:RCBF_set}
\end{proof}

\begin{remark}[Conservativeness] Theorem~\ref{theo:RCBF_set} requires that $(z_0,\dot e_0)\in S_V$, equivalently
$h_V(z_0,\dot e_0)=-V(z_0,\dot e_0)+\alpha_e h(z_0)\ge 0,$
which implies
$h(z_0)\ge \frac{V(z_0,\dot e_0)}{\alpha_e}.$
Since the safe set $\s_{\mathrm{FoM}}$ is defined only by $h(z_0)\ge 0$, this condition is more conservative than merely requiring $z_0\in \s_{\mathrm{FoM}}$. We highlight, however, that a priori for underactuated systems it is not clear that the entire set $\mathcal{S}_\mathrm{FoM}$ is safe.
\end{remark}

\begin{remark}[Feasibility]
For systems with limited control authority, constructing a classical exponential tracking function may be difficult because monotone error decrease need not hold under admissible inputs. By contrast, our framework only requires a valid RTF and a nonempty induced set $S_V$. In particular, if the RoM safe set has nonempty interior and $V(z,0)=0$, then any $z$ satisfying $h(z)>0$ also satisfies $(z,0)\in S_V$, and thus $S_V\neq\emptyset$.
\end{remark}

\begin{theorem}[Safety Assessment]
\label{theo:safe_assess}
Let $V(z, \dot e)$ be the RTF over $S \subseteq \mathcal{Z} \times \R^n$ with rate $\beta$ and consider an initial state $x_0\in\s_{\mathrm{FoM}}$ of the FoM s.t. $(z_0, \dot e_0) \in S_V = \{(z, \dot e) \in \mathcal{Z} \times \R^{n}: h_{V}(z, \dot e) \geq 0\}$. Then, the solution $x_{\cl}(t)$ of the closed system always remain in safe state set $\s_{\mathrm{FoM}}$, i.e. $x_{\cl}(t)\in\s_{\mathrm{FoM}}$ $\forall t\geq0$. 
\end{theorem}

\begin{proof}
See Appendix~\ref{app:safe_assess}
\end{proof}

\begin{remark}[Recurrence vs. Safety]\label{rem:recurrence_safety}
Theorem~\ref{theo:safe_assess} reveals that although $h_V$ is a recurrent CBF and trajectories may temporarily leave the recurrent set $S_V$, they never leave the safe set $\s_{\mathrm{FoM}}$. The reason is that when $\beta>\alpha$, tracking deviations are corrected faster than the system can approach the safety boundary, ensuring $h(z_{\cl}(t))\ge 0$ for all time despite the recurrent behavior of $h_V$. This is illustrated in Fig.~\ref{fig:RTF_demo}: $V$ is recurrent, while $h$ remains strictly nonnegative.
\end{remark}

\vspace{-2em}
\begin{figure}[htbp]
    \centering
    \includegraphics[width=1.0\linewidth]{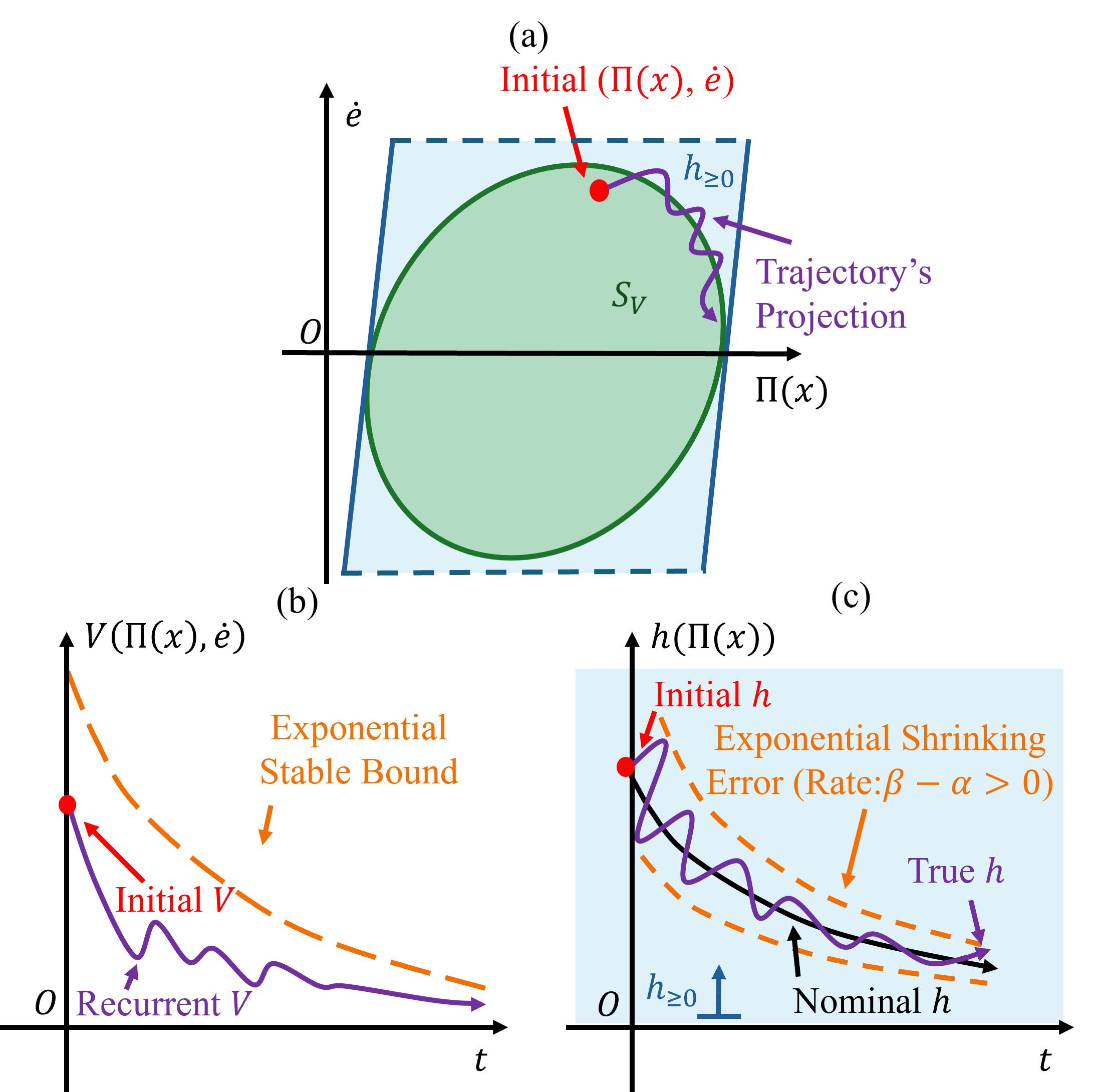}
    \caption{(a) $V, h, S_V$ and a safe trajectory's projection when $n =1$. (b) Time evolution of V. (c) Time evolution of $h$.}
    \label{fig:RTF_demo}
\end{figure}

\vspace{-1em}
\subsection{Effect of Disturbances}
In practice, feedback controllers $K(x,k(\Pi(x)))$ and $k(\Pi(x))$ often cannot achieve ideal exponential tracking due to model uncertainties, actuation limits, and external disturbances. We model this mismatch via a bounded disturbance $d$, under which the tracking error becomes Input-to-State Stable (ISS): 
\begin{align}
\|\dot{e}(t)\| \le M \|\dot{e}_0\| e^{-\beta t} + \mu(\|d\|_{\infty})
\end{align}
for some class-$\mathcal{K}$ function $\mu(\cdot)$. The key insight is that while disturbances prevent the tracking error from vanishing completely, they introduce only a bounded steady-state offset $\mu(\|d_{\infty}\|)$.

To accommodate this, we modify the RTF condition~\eqref{eq:exp_RTF} to account for the disturbance-induced offset:
\begin{align}
\label{eq:practical_RTF}
\min_{t \in T_{S}(x; \tau)} & e^{\beta t} (V(z_{\cl}(t), \dot{e}(t)) - \iota(\|d\|_{\infty})) \nonumber\\
& \le V(z_0, \dot{e}_0) - \iota(\|d\|_{\infty}),
\end{align}
where $\iota(\|d\|_{\infty}) = \frac{a_2e^{\beta \tau}\mu(\|d\|_{\infty})}{M}$~\cite[Theorem 1]{siegelmann2025data} quantifies the impact of disturbances on the RTF.

Besides, considering the tracking disturbance, rather than enforcing safety, we require Input-to-State Safety (ISSf), namely, the forward invariance of an enlarged set $S_d \supseteq h_{\ge 0}$ defined by
\begin{align}
S_d & = \{\, z \in \mathcal{Z} : h_d(z) \ge 0 \}, \\ 
h_d(z) & = h(z) + \gamma(\|  d \|_\infty)/\alpha_e,
\end{align}
where $\gamma$ is a class-$\mathcal{K}$ function to be specified. This leads to a robust safety guarantee through an enlarged recurrent set:
\begin{align}
    & S_{Vd} = \{(z, \dot{e}) \in \mathcal{Z} \times \R^{n}: h_{Vd}(z, \dot{e}) \ge 0\},\\
    & h_{Vd}(z, \dot{e}) = h_{V}(z, \dot{e}) + \gamma(\|d\|_{\infty}),
\end{align}
where $\gamma(\|d\|_{\infty}) = \frac{(2\beta - \alpha)\iota(\|d\|_{\infty})}{\alpha}$ provides a safety margin that scales with the disturbance magnitude. The following corollary formalizes this ISSf property.

\begin{corollary}[Input-to-State Safety]\label{cor:ISS_safety}
For system~\eqref{eq:closed_loop_FoM} with $\beta > \alpha$, if an initial state $x_0 \in \s_{\mathrm{FoM}}$ satisfies $(z_0, \dot{e}_0) \in S_{Vd}$ and the ISS tracking condition~\eqref{eq:practical_RTF} holds, then the ISSf for all $t \geq 0$, where $\gamma(\|d\|_{\infty}) = (2\beta -\alpha)\iota(\|d\|_{\infty})/\alpha$ and $\alpha_e = \frac{a_1^2(\beta - \alpha)}{a_2C_{h} M}$ is guaranteed.
\end{corollary}

\begin{proof}
See Appendix~\ref{app:ISSf}.
\end{proof}

\section{Case Study}
\label{sec:simulations}
In this section, we use a 2D double-integrator example as an illustrative proof-of-concept to visualize the role of RTFs and the induced recurrent safe set $S_V$. Consider the dynamics: $\ddot z = u,$ where $z \in \R^2$ is the agent's position, the FoM's dimension is 4, the RoM is $\dot z =v$, and $u \in \R^2$. Our goal is to navigate the system from a start position $z_0$ to a goal $z_g$ while avoiding obstacles. A simple solution is to realize the desired velocity $\dot z_d = -K_p (z - z_g)$ which is a proportional controller with gain $K_p \in \R_{>0}$.

In our setting, we consider a planar projection $\Pi(\X) \subseteq \R^2$ occupied with $N$ circular (closed-disk) obstacles. The $i$-th obstacle is $\Oi = \{z \in \R^2 : \|z - o_i\|\le r_i \}, i\le N,$ whose center is $o_i$ and radius is $r_i$. Then the obstacle's configuration space is the union $\Of = \cup_{i=1}^{N} \Oi$, and the collision-free configuration space is $\Pi(\X)/\Of$. The CBF we used for the RoM is $h(z): = \min_{i \in [1,N]\subset \mathbb{N}_{+}} \|z - o_i\| - r_i$, and its gradient $\nabla h(z) = \frac{(z-o_i)^{\top}}{\|z- o_i\|} = n_i^\top$ is an unit vector pointing from the nearest obstacle $\Oi$ to the agent. Then the safety velocity $\dot z_s$ w.r.t. $\Oi$ is given by the following quadratic program:
\begin{align}
    \arg & \min_{\dot z_s \in \R^{2}} (\dot z_s - \dot z_d)^{\top}(\dot z_s - \dot z_d)\\
    \mathrm{s.t.} & n_i^\top \dot z_s \ge - \alpha (\|z- o_i\| - r_i),
\end{align}
and its solution is $\dot z_s = \dot z_d + \max\{-n_i^\top \dot z_d - \alpha(\|z- o_i\| - r_i), 0\}\cdot n_i$~\cite{molnar2021model}. Then the safe velocity tracking controller can be defined as $u = -K_D(\dot z - \dot z_s)$ with $K_D > 0$. 

\begin{figure}[htbp]
\centering
\includegraphics[width=1.0\linewidth]{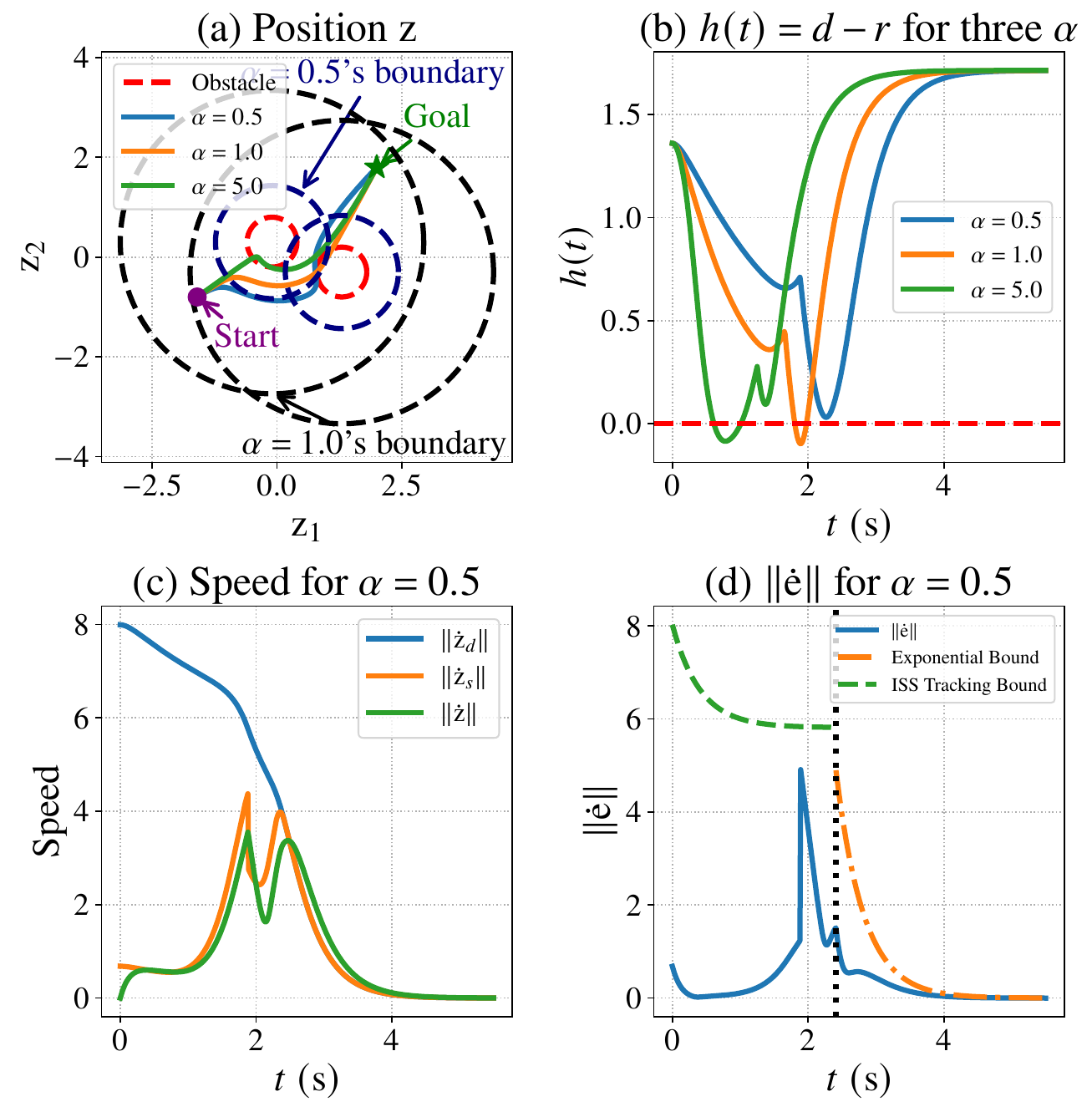}
\caption{(a) 2D path with circular obstacles. (b) Barrier value $h$ vs. time for the three $\alpha$. (c) Speeds for $\alpha=0.5$ of $\|\dot z_d\|$, $\|\dot z_s\|$, and $\|\dot z\|$. (d) Tracking-error speed $\|\dot e\|$ for $\alpha=0.5$.}
\label{fig:2D_DI}
\end{figure}

We illustrate three different projected trajectories with three different choices of $\alpha = 0.5, 1, 5$ respectively, and $K_P = 1.8, K_D = 8$. 
Using standard linear control design and~\cite{molnar2021model}, we can show that when the CBF constraint is inactive, this choice of parameters leads to an exponentially convergent tracking error with  $\beta = 2.45$ and $M = 3.24$. This allows us to characterize the set $S_V$ of safe initial conditions.
As shown in Fig~\ref{fig:2D_DI} (a), the red dashed lines show the boundary of the obstacles $\{z \in \R^2 | \|z - [-0.1, 0.3]^\top\| \le 0.5 \cup \|z - [1.3, -0.3]^\top\| \le 0.5\},$ the union of blue dashed circles denotes the unsafe initial states when $\alpha = 0.5$, and the black ones denote the unsafe initial states when $\alpha = 1$. As Fig.~\ref{fig:2D_DI} (a) and (b) show, when the CBF for RoM is invalid ($\alpha$ = 5) or the initial state is outside $S_V$ ($\alpha$ = 1), the system's safety cannot be guaranteed, and the minimum of $h(t)$ is less than 0. Only when Theorem~\ref{theo:RCBF_set} and Theorem~\ref{theo:safe_assess} are satisfied with a valid CBF in the RoM, can we ascertain that the trajectory is safe all the time ($\alpha = 0.5$). Besides, Fig.~\ref{fig:2D_DI} (c) and (d) verified that the tracking error shrink to 0 exponentially fast when~\eqref{eq:exp_RTF} is satisfied, and the tracking error is bounded by the ISS tracking bound computed by the $\|d\|_{\infty} = \|\ddot z_s\|_{\infty}$ when the safety filter is activated. 

\section{Conclusion and Future Work}
\label{sec:conclusion}

We developed a recurrence-based framework for layered safety-critical control of high-order nonlinear systems. By introducing Recurrent Tracking Functions (RTFs), we relax the classical Lyapunov-based tracking requirement to a finite-time recurrence condition, thereby allowing transient deviations while still preserving safety. In particular, any norm of an exponentially stable tracking error can serve as an RTF, yielding a certificate applicable to general nonlinear systems. We showed that augmenting RoM-based CBFs with RTFs leads to recurrent CBFs of the form $h_V(z, \dot{e}) = -V(z, \dot{e}) + \alpha_e h(z)$, which guarantee safety for all initial conditions in $S_V = \{h_V \ge 0\}$. The condition $\beta > \alpha$ ensures that the tracking error converges faster than the system approaches the safety boundary, thereby maintaining $h(z_{\cl}(t)) \ge 0, \forall t \ge 0,$ despite the recurrent nature of $h_V$. Numerical simulations further validate the theoretical results.

Future work will build on recent advances in data-driven verification of recurrent sets~\cite{siegelmann2025data,sbm2022cdc} to develop learning-based approaches for constructing RTFs directly from trajectory data, eliminating the need for explicit tracking function design. Experimental validation on physical robotic platforms also remains an important direction for future research.

\section{Acknowledgments}
The authors would like to thank Aaron Ames for suggesting this problem as an application of recurrence-based analysis.

\printbibliography
\vspace{-2.5mm}
\appendix
\subsection{Proof of Theorem~\ref{theo:RCBF_set}}\label{app:RCBF_set}
\begin{proof}
For simplicity let $T$ denote $T_{S}((z,\dot e); \tau)$. To prove the Theorem~\ref{theo:RCBF_set}, for any $x\in \s_{\mathrm{FoM}}$ s.t. $(\Pi(x),\dot e)\in S_V$, our goal is to find a feedback control signal $u(\cdot) = K(x(\cdot), k(\Pi(x(\cdot)))) \in \U^{(0, \infty]}$ and $\gamma:\R \rightarrow \R_{>0}$ such that:
\begin{align}
\label{eq:formulate_RCBF_h_V}
    \max_{t\in T} \; e^{\gamma(h_{V}(z(t), \dot e(t)))} h_{V}(z(t),\dot e(t)) \ge h(z, \dot e).
\end{align}

Then we can prove, when $u(\cdot)$ is the feedback controller that satisfies~\eqref{eq:exp_RTF} and $\gamma(h_{V}(z(t))) = \alpha t$, \eqref{eq:formulate_RCBF_h_V} is satisfied and $S_{V}$ is the corresponding control $\tau$-recurrent set. 

For the trajectory $x(\cdot)$ generated by the feedback controller that satisfies~\eqref{eq:exp_RTF}, we have 
\begin{align}
    \max_{t\in T} & \; e^{\alpha t}h_{V}(z(t), \dot e(t)) - h_{V}(z, \dot e)\\
    = \max_{t \in T} &  \int_{0}^{t} \frac{d(e^{\alpha s}h_{V}(z(s), \dot e(s)))}{ds} ds \label{theo:fundamental_theorem_of_calculus}\\
    = \max_{t \in T} &  \int_{0}^{t}e^{\alpha s}(-\dot{V}(z(s), \dot e(s)) \dots \nonumber\\
    & \quad + \alpha_e \nabla h ^\top (\dot{z}_{s} + \dot e) + \alpha h_{V}(z(s), \dot e(s))) ds \label{theo:definition_of_h_V}\\
    \ge \max_{t \in T} &  \int_{0}^{t} e^{\alpha s}(-\dot{V}(z(s), \dot e(s)) 
    - \alpha_{e} \alpha h(z(s)) \dots \nonumber\\
    & \quad - \alpha_{e} \|\nabla h\|\|\dot e\| + \alpha h_{V}(z(s), \dot e(s)))ds \label{theo:recrrent_cauchy}\\
\end{align}
\begin{align}
    \ge \max_{t \in T} &  \int_{0}^{t} e^{\alpha s}\!(-\dot V(z(s), \dot e(s)) - (\alpha+\frac{\alpha_e C_h}{a_1})\dots \nonumber\\
    & \cdot V(z(s), \dot e(s)))ds \label{theo:bound_RTF}\\
    = \max_{t \in T} &  \int_{0}^{t} - e^{-\frac{\alpha_e C_h}{a_1} s} \cdot\frac{d (e^{(\alpha + \frac{\alpha_e C_h}{a_1}) s}V(z(s), \dot e(s)))}{ds}ds\label{theo:integration_by_parts}\\
    = \max_{t \in T} &  V(z, \dot e) - e^{\alpha t} V(z(t), \dot e(t)) \dots \nonumber\\
    & \quad + \int_{0}^{t} \frac{d(e^{-\frac{ \alpha_e C_h}{a_1} s})}{ds} \, e^{(\alpha+\frac{\alpha_e C_h}{a_1})s}  V(z(s), \dot e(s)) \, ds \\
    = \max_{t \in T} &  V(z, \dot e) - e^{\alpha t} V(z(t), \dot e(t)) \dots \nonumber\\
    & \quad - \frac{\alpha_e C_h}{a_1} \int_{0}^{t} e^{\alpha s} V(z(s), \dot e(s)) ds\\
    \ge \max_{t \in T} &  V(z, \dot e) - e^{\alpha t} V(z(t), \dot e(t)) \dots \nonumber\\
    & \quad - \frac{a_2 \alpha_e C_h}{a_1} \int_{0}^{t} M e^{-(\beta - \alpha) s} \|\dot e\| \, ds \label{theo:recurrence_lineara2_1}\\
    \ge \max_{t \in T} &  [1 - \frac{a_2 \alpha_e C_h M}{a_1^2 (\beta - \alpha)}(1 - e^{-(\beta - \alpha) t})] V(z, \dot e) \dots \nonumber \\
    & \quad - e^{\alpha t} V(z(t), \dot e(t)) \label{theo:recurrence_lineara1_2}\\
    = \max_{t \in T} &  [1 - \frac{a_2 \alpha_e C_h M}{a_1^2(\beta - \alpha)}(1 - e^{-(\beta - \alpha) t})] \dots \nonumber\\
    & \quad \cdot [V(z, \dot e) - e^{\beta t} V(z(t), \dot e(t))], \label{theo:alpha_e} 
\end{align}
where the equality~\eqref{theo:fundamental_theorem_of_calculus} holds from the Fundamental Theorem of Calculus, the equality~\eqref{theo:definition_of_h_V} follows from the principle of integration by parts, the equality~\eqref{theo:integration_by_parts} is equavalent to~\eqref{eq:h_V}, and the equality~\eqref{theo:alpha_e} holds since $\alpha_e = \frac{a_1^2(\beta - \alpha)}{a_2C_{h} M}$. The inequality~\eqref{theo:recrrent_cauchy} can be obtained from Cauchy Inequality and 
\[
\nabla h(z)^{\top}\dot z_s \ge -\alpha h(z),
\]
the inequality~\eqref{theo:bound_RTF} and the inequality~\eqref{theo:recurrence_lineara1_2} can be obtained from the linear lower bound of the tracking error $a_1 \|\dot e\| \le V(\dot e)$, and the inequality~\eqref{theo:recurrence_lineara2_1} is equvalent to the upper bound of the tracking error $V(\dot e)\le a_2 \|\dot e\|$.

Note that $e^{-(\beta - \alpha)t} - e^{-(\beta - \alpha)\tau} \ge 0, \forall t \in T$, since $\beta > \alpha$ and $\min\limits_{t\in T} e^{\beta t}V(z(t), \dot e(t)) - V(z, \dot e) \le 0$, i.e. $\max\limits_{t\in T}V(z, \dot e) - e^{\beta t}V(z(t), \dot e(t)) \ge 0$,then we have
\begin{align}
& \max\limits_{t\in T} [e^{-(\beta - \alpha)t} - e^{-(\beta - \alpha)\tau}] \dots \\
& \quad \cdot[V(z, \dot e) - e^{\beta t}V(z(t), \dot e(t))]\\
\ge &  [e^{-(\beta - \alpha)t} - e^{-(\beta - \alpha)\tau}] \dots \\
 & \quad \cdot \max\limits_{t\in T} [V(z, \dot e) - e^{\beta t}V(z(t), \dot e(t))]\label{theo:max}\\
\ge &  0,\\
\end{align}
where~\eqref{theo:max} holds naturally from the property of the maximum,
\begin{align}
\text{i.e., } &  \eqref{theo:alpha_e} \nonumber\\
\ge &  e^{-(\beta - \alpha) \tau} \max_{t \in T} [V(z, \dot e) - e^{\beta t} V(z(t), \dot e(t))] \ge  0, 
\end{align}
thus,
\begin{align}
        \max_{t\in T} & \; e^{\alpha t}h_{V}(z(t), \dot e(t)) \ge h_{V}(z, \dot e),
\end{align}
and $S_V$ is the corresponding control $\tau$-recurrent set~\cite[Theorem 2]{liu2025RCBF}, which completes the proof.
\end{proof}

\subsection{Proof of Theorem~\ref{theo:safe_assess}}
\label{app:safe_assess}
\begin{proof}
To prove the Theorem~\ref{theo:safe_assess}, note that
states in $S_V$ impose the pointwise budget inequality $\frac{V(z,\dot e)}{\alpha_e} \le h(z).$
We know that $\nabla h(z)^{\top}\dot z_s \ge -\alpha h(z)$ and $\|\nabla h(z)\|\le C_h$. Along the FoM trajectory $x(\cdot)$, we have 
\begin{align} 
& e^{\alpha t} h(z(t)) - h(z) \\
= & \int_{0}^{t} \frac{d}{ds}(e^{\alpha s} h(z(s))) ds \label{theo:fundamental}\\
= & \int_{0}^{t} e^{\alpha s} (\dot h(z(s)) + \alpha h(z(s))) ds \label{theo:Leibniz_rule}\\
= & \int_{0}^{t} e^{\alpha s} (\nabla h(z(s))^{\top}(\dot z_s + \dot e(s)) + \alpha h(z(s))) ds\\
\ge & \int_{0}^{t} e^{\alpha s} (- \alpha h(z(s)) + \nabla h(z(s))^{\top}\dot e(s) + \alpha h(z(s)))ds \label{theo:safe_rom}\\
= & \int_{0}^{t} e^{\alpha s} \nabla h(z(s))^{\top} \dot e(s) ds\\
\ge & - C_h \int_{0}^{t} e^{\alpha s} \|\dot e(s)\| ds, \label{theo:safe_cauchy}
\end{align}
where the equality~\eqref{theo:fundamental} follows from the Fundamental Theorem of Calculus, and the equality~\eqref{theo:Leibniz_rule} holds naturally from Leibniz rule. Inequalities~\eqref{theo:safe_rom} and~\eqref{theo:safe_cauchy} follow from $\nabla h(z)^{\top}\dot z_s \ge -\alpha h(z)$ and Cauchy Inequality, respectively. Then $\forall t$, we have 
\begin{align}
h(z(t)) \ge e^{-\alpha t} h(z)- C_h \int_{0}^{t} e^{-\alpha (t-s)} \|\dot e(s)\|ds.
\end{align}
To be further, $\forall t \ge 0$, we have
\begin{align}
 & h(z(t)) \nonumber \ge e^{-\alpha t} h(z)- C_h \int_{0}^{t} e^{-\alpha (t-s)} \|\dot e(s)\|ds \\
 \ge &  e^{-\alpha t} h(z) - C_h M e^{-\alpha t} \int_{0}^{t} e^{-(\beta - \alpha) s} \|\dot e\|ds \label{theo:safe_exponential}\\
 = &  e^{-\alpha t} h(z)
-\frac{C_h M}{\beta-\alpha}e^{-\alpha t}(1-e^{-(\beta-\alpha)t})\|\dot e\| \\
 \ge &  e^{-\alpha t} \frac{V(z,\dot e)}{\alpha_e} - \frac{C_h M}{\beta-\alpha}e^{-\alpha t}(1-e^{-(\beta-\alpha)t})\|\dot e\| \label{theo:initial_set}\\
 \ge &  e^{-\alpha t} V(z,\dot e) [\frac{1}{\alpha_e} - \frac{C_h M}{a_1 (\beta - \alpha)} (1 - e^{-(\beta- \alpha)t})] \label{theo:last_linear}\\
 \ge &  e^{-\alpha t} V(z,\dot e) [\frac{1}{\alpha_e} - \frac{C_h M}{a_1 (\beta - \alpha)}] \label{theo:1}\\
 = &  e^{-\alpha t} V(z,\dot e) \frac{C_h M}{a_1 (\beta - \alpha)} (\frac{a_2}{a_1} - 1)\\
 \ge &  0, \label{theo:safe_linear_contained}
\end{align}
where the inequality~\eqref{theo:safe_exponential} holds from Theorem~\ref{theo:exponential_stability}, i.e., the exponential stability of tracking error under recurrence condition, the inequality~\eqref{theo:initial_set} holds since $(z,\dot e) \in S_V$, the inequality~\eqref{theo:last_linear} is obtained from the linear lower bound of the tracking error $a_1 \|\dot e\| \le V(z,\dot e)$, and the inequality~\eqref{theo:safe_linear_contained} follows from the fact that $a_1 \le a_2$ and $\beta > \alpha$. 

Since $h(\Pi(x(t))) \ge 0, \forall t \ge 0$, then we have $x(t) \in \s_{\mathrm{FoM}}, \forall t \ge 0$, i.e., the trajectory is always in the safe region, which completes the proof.
\end{proof}

\subsection{Proof of Corollary~\ref{cor:ISS_safety}}
\label{app:ISSf}
\begin{proof}
For simplicity, we use the same $T$ denoted in the Theorem~\ref{theo:RCBF_set} and $\gamma, \alpha, \mu$ denote $\gamma(\|d\|_{\infty}), \alpha(\|d\|_{\infty}), \mu(\|d\|_{\infty})$ respectively. For the trajectory $x(t)$ generated by the feedback controller $K(x, k(\Pi(x)))$ that satisfies~\eqref{eq:exp_RTF}, we have 
\begin{align}
    \max_{t\in T} & \; e^{\alpha t}(h_{V}(z(t), \dot e(t)) + \gamma) - (h_{V}(z, \dot e)+\gamma)\\
    = \max_{t \in T} &  \int_{0}^{t} \frac{d(e^{\alpha s}(h_{V}(z(s), \dot e(s))+\gamma))}{ds} ds\\
    = \max_{t \in T} &  \int_{0}^{t}e^{\alpha s}(-\dot{V}(z(s), \dot e(s)) + \alpha_e \nabla h ^\top (\dot{x}_{s} + \dot e) \dots \nonumber\\
    & \quad + \alpha (h_{V}(z(s), \dot e(s)) + \gamma)) ds\\
    \ge \max_{t \in T} &  \int_{0}^{t} e^{\alpha s}(-\dot{V}(z(s), \dot e(s)) 
    - \alpha_{e} \alpha h(z(s)) \dots \nonumber\\
    & \quad - \alpha_e \|\nabla h\|\|\dot e\| + \alpha (h_{V}(z(s), \dot e(s)) + \gamma))ds \label{coro:recrrent_cauchy}\\
    \ge \max_{t \in T} & \int_{0}^{t} e^{\alpha s}(-\dot V(z(s), \dot e(s)) - (\alpha + \frac{\alpha_e C_h}{a_1}) \dots \nonumber\\
    & \cdot (V(z(s), \dot e(s))-\gamma'))ds \label{coro:bound_RTF}\\
    = \max_{t \in T} &  \int_{0}^{t} - e^{-\frac{\alpha_e C_h}{a_1} s} \frac{\mathrm{d} (e^{(\alpha + \frac{\alpha_e C_h}{a_1}) s}V(z(s), \dot e(s)))}{\mathrm{d}s} \dots \nonumber\\
    + & \gamma' e^{-\frac{\alpha_e C_h}{a_1} s}\frac{\mathrm{d} e^{(\alpha + \frac{\alpha_e C_h}{a_1}) s}}{\mathrm{d}s} \mathrm{d}s\label{coro:recurrence_lineara1_1} \\
    = \max_{t \in T} &  V(z, \dot e) - e^{\alpha t} V(z(t), \dot e(t)) + (e^{\alpha t} -1)\gamma' \dots \nonumber\\
    + \int_{0}^{t} & \frac{d(e^{-\frac{\alpha_e C_h}{a_1} s})}{ds} e^{(\alpha+\frac{\alpha_e C_h}{a_1}) s} (V(z(s), \dot e(s)) - \gamma')ds \\
    = \max_{t \in T} & V(z, \dot e) - e^{\alpha t} V(z(t), \dot e(t)) + (e^{\alpha t} -1)\gamma' \dots \nonumber\\
    + \int_{0}^{t} & \frac{d(e^{-\frac{\alpha_e C_h}{a_1} s})}{ds} e^{(\alpha+\frac{\alpha_e C_h}{a_1}) s} V(z(s), \dot e(s))ds \dots \nonumber\\
    + & \frac{\alpha_e C_h}{a_1} \int_{0}^{t} e^{\alpha s} \gamma' ds\\
    = \max_{t \in T} &  V(z, \dot e) - e^{\alpha t} V(z(t), \dot e(t)) + (e^{\alpha t} -1)\gamma' \dots \nonumber\\
    - \frac{\alpha_e C_h}{a_1} & \int_{0}^{t} e^{\alpha s} V(z(s), \dot e(s))ds + \frac{\alpha_e C_h}{a_1 \alpha} (e^{\alpha t} - 1) \gamma'\\
    = \max_{t \in T} &  V(z, \dot e) - e^{\alpha t} V(z(t), \dot e(t)) \dots \nonumber\\
    - \frac{\alpha_e C_h}{a_1} & \int_{0}^{t} e^{\alpha s} V(z(s), \dot e(s)) ds + (e^{\alpha t} -1)\gamma
\end{align}
\begin{align}
    \ge \max_{t \in T} & V(z, \dot e) - e^{\alpha t} V(z(t), \dot e(t)) \dots \nonumber\\
    - \frac{a_2 \alpha_e C_h}{a_1} & \int_{0}^{t} (M e^{-(\beta - \alpha) s} \|\dot e\| + e^{\alpha s} \mu) ds + (e^{\alpha t} - 1) \gamma \label{coro:recurrence_lineara2_1}\\
    \ge \max_{t \in T} &  [1 - \frac{a_2 \alpha_e C_h M}{a_1^2 (\beta - \alpha)}(1 - e^{-(\beta - \alpha) t})] V(z, \dot e) \dots \nonumber \\
    - & e^{\alpha t} V(z(t), \dot e(t)) + (e^{\alpha t} -1)(\gamma - \frac{a_2 \alpha_e C_h}{a_1 \alpha} \mu) \label{coro:recurrence_lineara1_2}\\
    = \max_{t \in T} & [1 - \frac{a_2 \alpha_e C_h M}{a_1^2(\beta - \alpha)}(1 - e^{-(\beta - \alpha) t})] \dots \nonumber\\
    \cdot [ & V(z, \dot e) - e^{\beta t} V(z(t), \dot e(t))] + (e^{\alpha t} -1) \dots \nonumber\\
     \cdot & (\gamma - \frac{a_2 \alpha_e C_h}{a_1 \alpha} \mu), \\
     = \max_{t \in T} & e^{-(\beta - \alpha)t} [V(z, \dot e) - \iota - e^{\beta t} (V(z(t), \dot e(t))-\iota)] \dots \nonumber\\
    - & e^{-(\beta-\alpha)t} (e^{\beta t} -1)\iota+(e^{\alpha t} -1)(\gamma - \frac{a_2 \alpha_e C_h}{a_1 \alpha} \mu) \label{coro:recurrence_sign_preserve}
\end{align}
where the inequality~\eqref{coro:recrrent_cauchy} can be obtained from Cauchy Inequality and $\nabla h(z)^{\top}\dot z_s \ge -\alpha h(z)$, the inequality~\eqref{coro:bound_RTF}, and the inequality~\eqref{coro:recurrence_lineara1_2} can be obtained from the linear lower bound of the tracking error $a_1 \|\dot e\| \le V(z,\dot e)$, where 
\[
\gamma' = \frac{\alpha a_1}{\alpha a_1 + \alpha_e C_h} \gamma,
\]
and the inequality~\eqref{coro:recurrence_lineara2_1} is equvalent to the upper bound of the tracking error $V(z,\dot e)\le a_2 \|\dot e\|$.

Note that $e^{-(\beta - \alpha)t} - e^{-(\beta - \alpha)\tau} \ge 0, \forall t \in T$, since $\beta > \alpha$ and 
\[
\min\limits_{t\in T} e^{\beta t} (V(z(t), \dot e(t)) - \iota) - (V(z, \dot e) - \iota) \le 0,
\]
then for $t\in T$, we have 
\begin{align}
\max \limits_{t\in T} & [e^{-(\beta - \alpha)t} - e^{-(\beta - \alpha)\tau}] \dots \\
& \quad \cdot[(V(z, \dot e) - \iota) - e^{\beta t}(V(z(t), \dot e(t)) - \iota)]\\
\ge &  [e^{-(\beta - \alpha)t} - e^{-(\beta - \alpha)\tau}] \dots \\
 & \quad \cdot \max\limits_{t\in T} [(V(z, \dot e) - \iota) - e^{\beta t}(V(z(t), \dot e(t))-\iota)] \label{coro:maximum}\\
\ge &  0,
\end{align}
where~\eqref{coro:maximum} holds from the property of the maximum and the non-negativity of $e^{-(\beta - \alpha)t} - e^{-(\beta - \alpha)\tau}, \forall t \in T$ and 
\[
\max\limits_{t\in T} (V(z, \dot e) - \iota) - e^{\beta t} (V(z(t), \dot e(t)) - \iota),
\]
\begin{align}
\text{i.e., } & \eqref{coro:recurrence_sign_preserve} \nonumber\\
\ge \max_{t \in T} & -e^{-(\beta-\alpha)t}(e^{\beta t} -1)\iota + (e^{\alpha t} -1) \dots \nonumber\\
\cdot & (\gamma - \frac{a_2 \alpha_e C_h}{a_1 \alpha} \mu)\\
= \max_{t \in T} & -e^{-(\beta - \alpha) t} (e^{\beta t} - 1) \frac{a_2 e^{\beta \tau }\mu}{M} \dots \nonumber\\
+( & e^{\alpha t} -1)(\gamma - \frac{a_1 (\beta - \alpha)}{M \alpha} \mu) \label{Delta}.
\end{align}
Note that \begin{align}
\gamma & = \frac{(2\beta - \alpha)}{\alpha}\iota\\
\ge & \iota(\frac{\beta}{\alpha} + \frac{(\beta - \alpha)}{\alpha}e^{-\beta \tau}), \label{coro:minus_exponential}\\
= & \max_{t \in T} \iota(\frac{1-e^{-\beta t}}{1-e^{-\alpha t}} + \frac{(\beta - \alpha)}{\alpha}e^{-\beta \tau}), \label{coro:monocity}
\end{align}
where~\eqref{coro:minus_exponential} follows from the fact that $e^{-\beta \tau} < 1$, $\iota \ge 0$, and $\beta > \alpha$, ~\eqref{coro:monocity} follows from the monotonic decreasing property of $\frac{1-e^{-\beta t}}{1-e^{-\alpha t}} + \frac{(\beta - \alpha)}{\alpha}e^{-\beta \tau}$ and $\lim\limits_{t\rightarrow 0} \frac{1-e^{-\beta t}}{1-e^{-\alpha t}} + \frac{(\beta - \alpha)}{\alpha}e^{-\beta \tau} = \frac{\beta}{\alpha} + \frac{(\beta - \alpha)}{\alpha}e^{-\beta \tau}$. 

Thus, for $t \in T$, we have 
\begin{align}
& \eqref{Delta} \nonumber \\
= & \max_{t \in T} -e^{-(\beta - \alpha) t} (e^{\beta t} - 1) \iota \dots \nonumber\\
+( & e^{\alpha t} -1)[\frac{(2\beta - \alpha)}{\alpha} \iota - \frac{a_1 (\beta - \alpha) e^{-\beta \tau}}{a_2 \alpha} \iota]\\
= & \max_{t \in T} (e^{\alpha t} - 1)[\frac{(2\beta - \alpha)}{\alpha} \iota - (\frac{a_1 (\beta - \alpha) e^{-\beta \tau}}{a_2 \alpha}\dots \nonumber\\
 & + \frac{1- e^{-\beta t}}{1-e^{-\alpha t }})\iota] \\
\ge & \max_{t\in T} (e^{\alpha t} - 1)[\frac{(2\beta - \alpha)}{\alpha} \iota - (\frac{a_1 (\beta - \alpha) e^{-\beta \tau}}{a_2 \alpha} \dots \nonumber\\
& + \frac{1- e^{-\beta t}}{1-e^{-\alpha t }})\iota] \\
\ge & \max_{t\in T}(e^{\alpha t} - 1)[\frac{(2\beta - \alpha)}{\alpha} \iota - (\frac{ (\beta - \alpha) e^{-\beta \tau}}{\alpha} \dots \nonumber\\
& + \frac{1- e^{-\beta t}}{1-e^{-\alpha t }})\iota]\label{coro:linear}\\
\ge & (e^{\alpha t} - 1) \max_{t\in T}[\frac{(2\beta - \alpha)}{\alpha} \iota - (\frac{ (\beta - \alpha) e^{-\beta \tau}}{\alpha} \dots \nonumber\\
& + \frac{1- e^{-\beta t}}{1-e^{-\alpha t }})\iota]\label{coro:in}\\
= & (e^{\alpha t} - 1) [\frac{(2\beta - \alpha)}{\alpha} \iota - \min_{t\in T}(\frac{(\beta - \alpha) e^{-\beta \tau}}{\alpha} \dots \nonumber\\
& + \frac{1- e^{-\beta t}}{1-e^{-\alpha t }})\iota]\\
\ge & (e^{\alpha t} - 1) [\frac{(2\beta - \alpha)}{\alpha} \iota - \max_{t\in T}(\frac{(\beta - \alpha) e^{-\beta \tau}}{\alpha} \dots \nonumber\\
& + \frac{1- e^{-\beta t}}{1-e^{-\alpha t }})\iota]\\
& \ge 0 \label{coro:gamma},
\end{align}
where~\eqref{coro:linear} follows from $a_1\le a_2$ and $\beta > \alpha$,~\eqref{coro:in} and~\eqref{coro:gamma} follow from~\eqref{coro:monocity} and $e^{\alpha t} - 1 \ge 0$ for $t \ge 0$.

States in $S_{Vd}$ impose the pointwise budget inequality $\frac{V(z,\dot e)-\gamma}{\alpha_e} \le h(z).$ We know that $\nabla h(z)^{\top}\dot z_s \ge -\alpha h(z)$ and $\|\nabla h(z)\|\le C_h$. Along the FoM trajectory $x(t)$, we have 
\begin{align}
& e^{\alpha t} h(z(t)) - h(z) \\
= & \int_{0}^{t} \frac{d}{ds}(e^{\alpha s} h(z(s))) ds \\
= & \int_{0}^{t} e^{\alpha s} (\dot h(z(s)) + \alpha h(z(s))) ds\\
= & \int_{0}^{t} e^{\alpha s} (\nabla h(z(s))^{\top}(\dot z_s + \dot e(s)) + \alpha h(z(s))) ds\\
\ge & \int_{0}^{t} e^{\alpha s} (- \alpha h(z(s)) + \nabla h(z(s))^{\top}\dot e(s) + \alpha h(z(s)))ds \label{coro:safe_rom}\\
= & \int_{0}^{t} e^{\alpha s} \nabla h(z(s))^{\top} \dot e(s) ds \\
\ge & - C_h \int_{0}^{t} e^{\alpha s} \|\dot e(s)\| ds, \label{coro:safe_cauchy}
\end{align}
where~\eqref{coro:safe_rom} and~\eqref{coro:safe_cauchy} follow from $\nabla h(z)^{\top}\dot z_s \ge -\alpha h(z)$ and Cauchy Inequality, respectively.

Then $\forall t$, we have 
\begin{align}
h(z(t)) \ge e^{-\alpha t} h(z)- C_h \int_{0}^{t} e^{-\alpha (t-s)} \|\dot e(s)\|ds.
\end{align}
Accoring to the former proof of the Corollary~\ref{cor:ISS_safety}, we know that 
\[
\max_{t \in T} e^{\alpha t}(h_V(z(t), \dot e(t)) + \gamma) - (h_V(z, \dot e) + \gamma) \ge 0,\]
i.e. $S_{Vd}$ is a control recurrent set, which means that the trajectory will come back to $S_{Vd}$ every $\tau$ units of time. Thus, to prove 
\[
S_d = \{ z \in \mathcal{Z} : h_d(z) = h(z) + \gamma/\alpha_e \ge 0 \}
\]
is invariant, i.e., the ISSf is guaranteed. We only need to prove that every trajectory starts from $S_{Vd}$ will be contained in the set $S_{d}$ for at least $\tau$ units of time. The original parameter choices imply 
\[
\frac{\gamma}{\alpha_e}
= \frac{a_2^2}{a_1^2} \frac{2\beta-\alpha}{\beta-\alpha} e^{\beta\tau}\frac{C_h\mu}{\alpha} \ge \frac{C_h\mu}{\alpha},
\] 
since $a_1\le a_2$, $\beta>\alpha$, and $\mu\ge0$. Using the ISS tracking bound from Section~III-D, for every $t\in[0,\tau]$ we obtain the following pointwise estimate:
\begin{align}
& h(z(t)) \nonumber
\ge e^{-\alpha t}h(z)
-C_h\int_0^t e^{-\alpha(t-s)}\|\dot e(s)\|\,ds \nonumber\\
\ge & e^{-\alpha t}h(z)
-C_he^{-\alpha t}\int_0^t
\bigl(M e^{-(\beta-\alpha)s}\|\dot e\|
+e^{\alpha s}\mu\bigr)\,ds
\label{coro:safe_exponential}\\
= {}& e^{-\alpha t}h(z)
-\frac{C_hM}{\beta-\alpha}e^{-\alpha t}
(1-e^{-(\beta-\alpha)t})\|\dot e\| \nonumber\\
&\quad-\frac{C_h\mu}{\alpha}(1-e^{-\alpha t})\\
\ge {}& e^{-\alpha t}\frac{V(z,\dot e)-\gamma}{\alpha_e}
\nonumber\\
&\quad-\frac{C_hM}{\beta-\alpha}e^{-\alpha t}
(1-e^{-(\beta-\alpha)t})\|\dot e\| \nonumber\\
&\quad-\frac{C_h\mu}{\alpha}(1-e^{-\alpha t})\\
\ge {}& e^{-\alpha t}V(z,\dot e)
\left[\frac{1}{\alpha_e}
-\frac{C_hM}{a_1(\beta-\alpha)}
(1-e^{-(\beta-\alpha)t})\right] \nonumber\\
&\quad-e^{-\alpha t}\frac{\gamma}{\alpha_e}
-\frac{C_h\mu}{\alpha}(1-e^{-\alpha t})\\
\ge {}& e^{-\alpha t}V(z,\dot e)
\left[\frac{1}{\alpha_e}
-\frac{C_hM}{a_1(\beta-\alpha)}\right]
-\frac{\gamma}{\alpha_e}
\label{coro:monotonicity}\\
= {}& e^{-\alpha t}V(z,\dot e)
\frac{C_hM}{a_1(\beta-\alpha)}
\left(\frac{a_2}{a_1}-1\right)
-\frac{\gamma}{\alpha_e}\\
\ge {}&-\frac{\gamma}{\alpha_e}.
\label{coro:safe_linear_contained}
\end{align}
Here~\eqref{coro:safe_exponential} uses $\|\dot e(t)\|\le M\|\dot e\|e^{-\beta t}+\mu$. Inequality~\eqref{coro:monotonicity} follows from
$0\le1-e^{-(\beta-\alpha)t}\le1$ and $-e^{-\alpha t}\frac{\gamma}{\alpha_e} -(1-e^{-\alpha t})\frac{C_h\mu}{\alpha}
\ge -\frac{\gamma}{\alpha_e}.$ Finally,~\eqref{coro:safe_linear_contained} follows from
$a_1\le a_2$ and $\beta>\alpha$. Thus trajectories starting in $S_{Vd}$ remain in $S_d$ for the next $\tau$ time units. Applying this estimate at successive returns to $S_{Vd}$ establishes the claimed safety bound.
\end{proof}

\subsection{Domains and Return Conditions for the Safety Guarantees}
\label{app:clarifications}
Appendix~\ref{app:ISSf} extends its barrier bound through successive
returns to $S_{Vd}$. We specify the domains needed for this continuation,
establish recurrence at a common RTF witness, and give a sufficient
joint return condition for the norm-based construction. These
qualifications apply to Section~\ref{sec:RTF}: the tracking bound
controls excursions, whereas the RTF inequality supplies a return.
In Theorem~\ref{theo:exponential_stability}, joint return is an additional
hypothesis, not a consequence of tracking-error decay alone.

\subsubsection{Common setting and domains}
Write $q(t):=\dot e(t)$, $V_t:=V(z(t),q(t))$, and $h_t:=h(z(t))$.
For a prescribed disturbance bound $\|d\|_\infty\le\bar d$, abbreviate
$\mu:=\mu(\bar d)$, $\iota:=\iota(\bar d)$, and
$\gamma:=\gamma(\bar d)$. The nominal case is included by taking
$\mu=\iota=\gamma=0$. Throughout this appendix, $\alpha$ is the fixed
RoM CBF rate. Retain the parameter choices of the main text:
\begin{equation}
\label{app:eq:disturbed_parameters}
\begin{cases}
\delta := \beta-\alpha>0,\\
\alpha_e := \dfrac{a_1^2\delta}{a_2C_hM},\\
\iota := \dfrac{a_2e^{\beta\tau}\mu}{M},\\
\gamma := \dfrac{2\beta-\alpha}{\alpha}\iota.
\end{cases}
\end{equation}
The pointwise tracking hypothesis is
\begin{equation}\label{app:eq:iss_envelope}
\|q(t)\|\le M e^{-\beta t}\|q_0\|+\mu, \, t\ge0.
\end{equation}
For $\mu=0$, this is Assumption~\ref{ass:Exponential_Tracking}; for
$\mu>0$, it is the ISS bound in Section~III-D. It is used independently
of the recurrent inequalities~\eqref{eq:exp_RTF} and
\eqref{eq:practical_RTF}. In particular, the pointwise bound invoked
in~\eqref{theo:safe_exponential} comes from
Assumption~\ref{ass:Exponential_Tracking}, not from the converse
implication in Theorem~\ref{theo:exponential_stability}.
The rate $\beta>\alpha$ must be valid for both the tracking envelope
and the RTF condition. If a norm-based construction yields an RTF rate
$\beta'<\beta_{\mathrm{tr}}$, where $\beta_{\mathrm{tr}}$ is the actual
tracking rate, use $\beta=\beta'$ here and require $\alpha<\beta'$.
The tracking envelope remains valid at this smaller rate.

For the nominal closed loop, define the joint output map and the
compatible full-state domain by
\begin{align}
\Psi(x)&:=(\Pi(x),q(x)),\notag\\
q(x)&:=\dot z_{\cl}(x)-\dot z_s(\Pi(x)),
\label{app:eq:joint_output}\\
\mathcal X_S&:=\{x\in\mathcal X:\Psi(x)\in S\}.
\label{app:eq:compatible_domain}
\end{align}
Assume the induced nominal closed loop is well posed and forward complete,
$\Psi$ is continuous, and the safe RoM feedback has the regularity required
for its CBF invariance result. For disturbed trajectories, use the actual
joint pair $(z(t),q(t))$, assume $z$ is absolutely continuous, and impose
the tracking and barrier bounds for every compatible full-state initial
condition and admissible disturbance. A joint pair need not determine
a unique full state. The required initial pair in the containment-time
definition preceding Definition~\ref{def:RTF} is $(z_0,\dot e_0)$.

Since the linear bounds on $V$ are certified on $S$, the initial sets
used in the safety results are
\begin{align}
S_V&:=\{(z,q)\in S:V(z,q)\le\alpha_e h(z)\},\notag\\
S_{Vd}&:=\{(z,q)\in S:V(z,q)\le\alpha_e h(z)+\gamma\}.
\label{app:eq:restricted_sets}
\end{align}
Thus the superlevel sets in Theorem~\ref{theo:safe_assess} and
Corollary~\ref{cor:ISS_safety} are intersected with their verification
domain $S$. The initial requirement $h_0\ge0$ in the corollary is
retained. Similarly, the feasibility remark requires $(z,0)\in S$
for the chosen safe $z$.

The barrier estimate in Appendix~\ref{app:ISSf} must remain applicable
when $h$ is negative but within the allowed disturbance margin. Define
\begin{equation}\label{app:eq:enlarged_h_domain}
\Gamma:=\frac{\gamma}{\alpha_e}, \,\mathcal D_h(\bar d):=\{z\in\mathcal Z:h(z)\ge-\Gamma\}.
\end{equation}
A sufficient domain hypothesis is
\begin{equation}\label{app:eq:barrier_domain_conditions}
\|\nabla h(z)\|\le C_h,\, \nabla h(z)^\top\dot z_s(z)\ge-\alpha h(z),
\end{equation}
on $\mathcal D_h(\bar d)$, with the same $C_h$ as in
\eqref{app:eq:disturbed_parameters}. For $\mu=0$, this reduces to the
original safe domain. For $\mu>0$, Assumption~\ref{ass:boundness} alone
does not supply the bound on the enlarged domain. A smaller verified
domain is sufficient if it contains all relevant projected trajectories
up to a possible first exit from $\mathcal D_h(\bar d)$.

The endpoint arguments below use the lower bound on $V$ at initial
or restarted pairs in $S$. To use instead the substitutions involving
$V(z(t),q(t))$ inside the integrals of Appendices~\ref{app:RCBF_set}
and~\ref{app:ISSf}, the corresponding upper and lower bounds must
also hold at the intermediate pairs. In the nominal case, their
required domain is the joint reachable tube
\begin{equation}\label{app:eq:joint_tube}
\mathcal T_\tau(S):=
\bigcup_{\substack{x_0\in\mathcal X_S\\0\le t\le\tau}}
\{\Psi(x_{\cl}(t;x_0))\},
\end{equation}
with the analogous union over admissible disturbances in the disturbed
case. Compactness of $S$ does not itself extend the bounds to this tube.

\subsubsection{Continuation of the barrier estimate}
Under the preceding hypotheses, the calculation leading to
\eqref{coro:safe_cauchy} is valid up to a possible first exit from
$\mathcal D_h(\bar d)$. In the above notation it reads
\begin{equation}\label{app:eq:nominal_h_dot}
\dot h_t+\alpha h_t
\ge\nabla h(z(t))^\top q(t)\ge-C_h\|q(t)\|
\end{equation}
almost everywhere. Therefore
\begin{equation}\label{app:eq:nominal_h_integral}
e^{\alpha t}h_t-h_0
\ge-C_h\int_0^t e^{\alpha s}\|q(s)\|\,\md s.
\end{equation}
Let
\begin{equation}\label{app:eq:budget_constant}
B:=\frac{C_hM}{\delta}\|q_0\|.
\end{equation}
Substitution of~\eqref{app:eq:iss_envelope} gives the pointwise estimate
used in Appendix~\ref{app:ISSf}:
\begin{align}
h_t&\ge e^{-\alpha t}h_0
 -C_he^{-\alpha t}\int_0^t
 \bigl(M e^{-\delta s}\|q_0\|+e^{\alpha s}\mu\bigr)\,\md s\notag\\
&=e^{-\alpha t}\bigl[h_0-B(1-e^{-\delta t})\bigr]
 -\frac{C_h\mu}{\alpha}(1-e^{-\alpha t}).
\label{app:eq:disturbed_h_lower}
\end{align}
The disturbance contribution to this lower bound is negative.

The margin in~\eqref{app:eq:disturbed_parameters} satisfies
\begin{align}\label{app:eq:margin_bound}
\Gamma & =\frac{a_2^2}{a_1^2}
 \frac{2\beta-\alpha}{\beta-\alpha}
e^{\beta\tau}\frac{C_h\mu}{\alpha}\\
& \ge\frac{C_h\mu}{\alpha}.
\end{align}
For an initial pair in $S_{Vd}$, set $H_t:=h_t+\Gamma$. The certificate
budget and the lower bound on $V$ imply
\begin{equation}\label{app:eq:disturbed_budget}
H_0\ge\frac{V_0}{\alpha_e}\ge\frac{a_2}{a_1}B\ge B.
\end{equation}
Consequently,~\eqref{app:eq:disturbed_h_lower} yields
\begin{align}
H_t&\ge e^{-\alpha t}\bigl[H_0-B(1-e^{-\delta t})\bigr]\notag\\
&\quad+\left(\Gamma-\frac{C_h\mu}{\alpha}\right)(1-e^{-\alpha t})\notag\\
&\ge e^{-\beta t}H_0
 +\left(\Gamma-\frac{C_h\mu}{\alpha}\right)(1-e^{-\alpha t})
\ge0.
\label{app:eq:iss_safety}
\end{align}
If $H_0>0$, this bound rules out a first finite time at which $H_t=0$.
If $H_0=0$, the requirement $h_0\ge0$ implies $h_0=\Gamma=0$;
then $\mu=0$, $V_0=0$, and $q_0=0$. The tracking bound gives
$q(t)=0$, so the projected trajectory follows the safe RoM dynamics.
Thus, in either case,
\begin{equation}\label{app:eq:iss_all_time}
h(z(t))\ge-\gamma/\alpha_e,\qquad t\ge0.
\end{equation}
For $\mu=0$, this calculation recovers the nominal bound
\begin{equation}\label{app:eq:nominal_safety}
h_t\ge e^{-\beta t}h_0\ge0.
\end{equation}
Thus~\eqref{coro:safe_linear_contained} extends to all times, with a
first-exit justification for the gradient bound. This conclusion requires
the certificate~\eqref{app:eq:restricted_sets}; a bound on $h_0$ alone
does not certify arbitrary initial tracking errors.

\subsubsection{Recurrence at a common return time}
The preceding estimate establishes the barrier domain along the trajectory.
We can now justify the recurrence steps of
Appendices~\ref{app:RCBF_set} and~\ref{app:ISSf} on this domain.
Define
\begin{equation}\label{app:eq:rho}
\rho:=\frac{a_1}{a_2}\le1,\qquad
\frac{\alpha_e C_hM}{a_1\delta}=\rho.
\end{equation}
Use $h_{Vd}=-V+\alpha_e h+\gamma$, which reduces to $h_V$ in the
nominal case. Its endpoint difference is exactly
\begin{align}
\Delta_d(t)&:=e^{\alpha t}h_{Vd}(z(t),q(t))-h_{Vd}(z_0,q_0)\notag\\
&=V_0-e^{\alpha t}V_t+\alpha_e(e^{\alpha t}h_t-h_0)
 +(e^{\alpha t}-1)\gamma.
\label{app:eq:endpoint_identity}
\end{align}
This identity is valid for continuous $V$ and requires no time derivative
of $V$. With
\begin{equation}\label{app:eq:disturbance_c}
c_\mu:=\frac{\alpha_e C_h\mu}{\alpha}
=\rho^2\frac{\delta}{\alpha}e^{-\beta\tau}\iota,
\end{equation}
equations~\eqref{app:eq:nominal_h_integral} and
\eqref{app:eq:iss_envelope}, together with $a_1\|q_0\|\le V_0$, imply
\begin{align}
\Delta_d(t)&\ge[1-\rho(1-e^{-\delta t})]V_0-e^{\alpha t}V_t\notag\\
&\quad+(e^{\alpha t}-1)(\gamma-c_\mu)\notag\\
&\ge e^{-\delta t}[V_0-e^{\beta t}V_t]
 +(e^{\alpha t}-1)(\gamma-c_\mu).
\label{app:eq:positive_margin_term}
\end{align}
The margin term keeps its positive sign, consistently with
\eqref{coro:recurrence_lineara1_2}. The direct endpoint calculation uses
the coefficient $c_\mu$, avoiding linear bounds on $V$ outside $S$.

For $t>0$, define
\begin{equation}\label{app:eq:ratio_bound}
r(t):=\frac{1-e^{-\beta t}}{1-e^{-\alpha t}}
\le\frac{\beta}{\alpha}.
\end{equation}
The inequality follows from
\[
(1-e^{-kt})/k=\int_0^t e^{-ks}\,\md s,
\]
which decreases with $k>0$.
Using $e^{-\delta t}(e^{\beta t}-1)=(e^{\alpha t}-1)r(t)$ in
\eqref{app:eq:positive_margin_term} gives
\begin{align}
\Delta_d(t)&\ge e^{-\delta t}
 \bigl[V_0-\iota-e^{\beta t}(V_t-\iota)\bigr]\notag\\
&\quad+(e^{\alpha t}-1)[\gamma-c_\mu-r(t)\iota],
\label{app:eq:shifted_endpoint}
\end{align}
where
\begin{align}\label{app:eq:shifted_margin_nonnegative}
& \gamma-c_\mu-r(t)\iota \\
\ge & \frac{\delta}{\alpha}(1-\rho^2e^{-\beta\tau})\iota\\
\ge & 0.
\end{align}
Interpret the containment times in~\eqref{eq:practical_RTF} as
$T_S^e(x_0;\tau)$ along the actual joint trajectory. The shifted RTF
condition provides a time $t_r\in T_S^e(x_0;\tau)$ such that
\begin{equation}\label{app:eq:shifted_witness}
e^{\beta t_r}(V_{t_r}-\iota)\le V_0-\iota.
\end{equation}
At this same time, both terms on the right of
\eqref{app:eq:shifted_endpoint} are nonnegative. Therefore
\begin{equation}\label{app:eq:common_rcbf_witness}
e^{\alpha t_r}h_{Vd}(z(t_r),q(t_r))
\ge h_{Vd}(z_0,q_0)\ge0,
\end{equation}
and $(z(t_r),q(t_r))\in S_{Vd}$. When $\mu=0$, the shifted RTF
condition becomes~\eqref{eq:exp_RTF}, and the same calculation gives
the nominal RCBF inequality and return to $S_V$.

This common-time argument replaces the maximum-separation steps
in~\eqref{theo:max} and~\eqref{coro:maximum} and the extraction of
time-dependent factors from maxima: all terms are evaluated at one RTF
witness. The exponent function in Definition~\ref{def:RCBF} is the
constant $\alpha$; its full-state interpretation is specified below.

\subsubsection{Verifying the return for a norm-based RTF}
For the norm-based construction in
Theorem~\ref{theo:exponential_stability}, decay and joint return must
hold at the same time.

Suppose the nominal tracking bound
\begin{equation}\label{app:eq:tracking_converse}
\|q(t)\|\le M e^{-\beta_{\mathrm{tr}}t}\|q_0\|,
\qquad M\ge1,
\end{equation}
holds uniformly for $x_0\in\mathcal X_S$. For
$0<\beta'<\beta_{\mathrm{tr}}$, let
\begin{equation}\label{app:eq:minimum_return_time}
T_{\mathrm{dec}}:=\frac{\log M}{\beta_{\mathrm{tr}}-\beta'}.
\end{equation}
In addition to the hypotheses in the main text, require
\begin{equation}\label{app:eq:late_return}
\begin{gathered}
\forall x_0\in\mathcal X_S,\quad
\exists t_r\in[T_{\mathrm{dec}},\tau]\cap(0,\tau]\quad\text{such that}\\
\Psi(x_{\cl}(t_r;x_0))\in S.
\end{gathered}
\end{equation}
For $V(z,q):=\|q\|$, this time satisfies
\begin{align}\label{app:eq:norm_return_proof}
& e^{\beta't_r}V_{t_r}\\
\le & M e^{-(\beta_{\mathrm{tr}}-\beta')t_r}V_0 \\
\le & V_0.
\end{align}
Together with $a_1=a_2=1$, this proves the RTF condition with rate
$\beta'$ and horizon $\tau$. The displayed minimum in the definition
is also attained: continuity and closedness of $S$ imply that a
nonattained infimum could only arise as $t\downarrow0$, where the
weighted value tends to $V_0$. If the infimum is below $V_0$, it is
attained away from zero; otherwise~\eqref{app:eq:norm_return_proof}
provides a point attaining $V_0$.

A stronger but simple sufficient condition is
\begin{equation}\label{app:eq:fixed_horizon_return}
\Psi(x_{\cl}(\tau;x_0))\in S,\,\forall x_0\in\mathcal X_S \text{ and }\tau\ge \max\{T_{\mathrm{dec}},\tau>0\}.
\end{equation}
Positive invariance of $\mathcal X_S$ is sufficient as well. The less
restrictive condition~\eqref{app:eq:late_return} allows excursions
outside $S$ and an initial-condition-dependent return time. The horizon
in Theorem~\ref{theo:exponential_stability}, with $b_2\ge b_1>0$,
is at least $T_{\mathrm{dec}}$, but does not by itself establish
\eqref{app:eq:late_return}.

The distinction is illustrated by
\begin{equation}\label{app:eq:return_counterexample}
\dot z=1+q,\,\dot q=-q,\, S=[-1,1]^2.
\end{equation}
With reference $\dot z_s=1$ and initial pair $(1,0)$, the error bound
holds with $M=1$ and rate $1$, while $(z(t),q(t))=(1+t,0)$ lies
outside $S$ for every $t>0$. Thus a neighborhood of zero in the error
projection $S_{\dot e}$ cannot ensure a joint return.
In the full-state RLF converse, by comparison, convergence of
$\|x(t)-x^*\|$ and the inclusions
\[
B_r(x^*)\subseteq S\subseteq B_R(x^*)
\]
place the whole state back
in $S$ after sufficient time. The error projection alone does not
provide the corresponding implication for an RTF.

\subsubsection{Successive returns and the full-state interpretation}
We now specify the restart argument used at the end of
Appendix~\ref{app:ISSf}. In the nominal case, the tracking and RTF
hypotheses hold for every compatible initial state in $S$, and the
same autonomous closed loop generates the trajectory after each return.
For the disturbed case, require the tracking envelope and shifted RTF
condition to hold after every restart at a pair in $S_{Vd}$, uniformly
over admissible disturbance tails. Use the same fixed bound $\bar d$
and the same margins $\mu$, $\iota$, and $\gamma$ throughout. In
particular, a restart can occur with $h<0$: the domain needed there is
the enlarged set $\mathcal D_h(\bar d)$, not just the original safe set.
The all-time estimate~\eqref{app:eq:iss_all_time} already keeps the
trajectory in that enlarged domain. The endpoint argument can therefore
be reapplied at such returns without requiring $h\ge0$ anew.

Assume the joint trajectory is continuous when making a repeated-return
claim. From each return,~\eqref{app:eq:common_rcbf_witness} supplies
another visit to $S_{Vd}$ within $\tau$; the nominal case uses $S_V$.
These visits admit a sequence whose gaps are at most $\tau$ and whose
times diverge. To prove this, choose the latest membership time in
$(0,\tau]$ after each visit. It exists by continuity, closedness of
the joint certificate set, and the positive return witness. If the
chosen times accumulated at a finite $t_*$, forward completeness and
continuity would put the pair at $t_*$ in the same closed set. For
sufficiently late indices, $t_*$ would then be a later membership time
within the next horizon, contradicting the choice of the latest time.
The latest membership times need not themselves be the times witnessing
the weighted inequality. This establishes the recurrent-set assertion
and justifies the interval concatenation described in
Appendix~\ref{app:ISSf}. With the all-time
envelope~\eqref{app:eq:iss_envelope}, the direct continuation estimate
above gives the same safety bound without requiring that concatenation.

To relate the joint-coordinate inequality to Definition~\ref{def:RCBF},
consider the nominal full-state pullback
\begin{equation}\label{app:eq:pullback_barrier}
\widetilde h_V(x):=h_V(\Psi(x)).
\end{equation}
For a literal application of that definition on a compact full-state
domain, choose a compact $\widehat S\subseteq\mathcal X_S$ and require
the RTF witness times also to satisfy $x_{\cl}(t_r)\in\widehat S$.
Then the same inequality is an RCBF witness on
\[
\widehat S\cap\{\widetilde h_V\ge0\}.
\]
The analogous shifted
interpretation applies when the disturbed joint output has a continuous
full-state realization and the corresponding containment requirement
holds. Compactness of $S$ alone does not imply compactness of its
full-state preimage. The trajectory-wise safety estimates above apply
on their certified initial domain independently of this compact
full-state realization.
\end{document}